%% file: qpl_paper.tex
\documentclass[submission,copyright,creativecommons]{eptcs}
\providecommand{\event}{QPL 2023}

\usepackage{iftex}

\ifpdf
  \usepackage{underscore}         
  \usepackage[T1]{fontenc}        
  \usepackage[utf8]{inputenc}
\else
  \usepackage{breakurl}           
\fi

\usepackage{graphicx}
\usepackage{stfloats}
\usepackage{placeins}
\usepackage{float}
\usepackage{amsthm}
\usepackage{amsmath}
\usepackage{amssymb}
\usepackage{booktabs}
\usepackage{cite}

\usepackage[english]{babel}

\theoremstyle{definition}
\newtheorem{theorem}{Theorem}
\newtheorem{lemma}{Lemma}
\newtheorem{definition}{Definition}
\newtheorem{corollary}{Corollary}

\usepackage[utf8]{inputenc}
\usepackage[linesnumbered,ruled]{algorithm2e}

\usepackage{tikzit}
\input{quantum.tikzdefs}
\input{quantum.tikzstyles}

\title{Picturing Counting Reductions with the ZH-Calculus}

\author{Tuomas Laakkonen$^{1, a}$, Konstantinos Meichanetzidis$^{1, b}$, John van de Wetering $^{2, c}$
\institute{$^{1}$ Quantinuum, 17 Beaumont Street, Oxford OX1 2NA, United Kingdom\\
$^{2}$ Informatics Institute, University of Amsterdam, 1098 XH Amsterdam, The Netherlands}
\email{$^{a, b}$ \{tuomas.laakkonen, k.mei\}@quantinuum.com, $^{c}$ john@vdwetering.name}
}

\def\titlerunning{Picturing Counting Reductions with the ZH-Calculus}
\def\authorrunning{T. Laakkonen, K. Meichanetzidis \& J. van de Wetering}

\begin{document}
    \maketitle

\begin{abstract}
        Counting the solutions to Boolean formulae defines the problem \sSAT, which is complete for the complexity class \sP. We use the ZH-calculus, a universal and complete graphical language for linear maps which naturally encodes counting problems in terms of diagrams,
    to give graphical reductions from \sSAT to several related counting problems. Some of these graphical reductions, like to \stwoSAT, are substantially simpler than known reductions via the matrix permanent. Additionally, our approach allows us to consider the case of counting solutions modulo an integer on equal footing.
    Finally, since the ZH-calculus was originally introduced to reason about quantum computing, we show that the problem of evaluating scalar ZH-diagrams in the fragment corresponding to the Clifford+T gate set, is in $\textbf{FP}^\sP$. Our results show that graphical calculi represent an intuitive and useful framework for reasoning about counting problems.
\end{abstract}

Graphical calculi like the ZX-calculus~\cite{CD1,coecke_interacting_2011} are seeing increased usage in reasoning about quantum computations.
While earlier work in this area has mostly focused on \emph{representing} existing quantum protocols and quantum algorithms in a graphical way in order to shed light on how these protocols work~\cite{duncan2010rewriting,CES,hillebrand_superdense_2012,horsman2011quantum,duncan2013verifying,coecke_picturing_2017,gogioso2017fully,Vicary2013,shaikh2022sum}, recent years have seen the development of entirely new results that improve upon the existing state-of-the-art.
For instance, there are now new results proved with a graphical calculus in quantum circuit optimization~\cite{cliffsimp,kissinger2019tcount,deBeaudrapN2020treducspidernest,Cowtan2020phasegadget,cowtan2020generic,borgna2021hybrid}, verification~\cite{kissinger2019tcount,kostia2021geometry,lehmann2022vyzx}
and simulation~\cite{kissinger2021simulating,kissinger2022classical,laakkonen2022graphical,codsi2022classically,ufrecht2023cutting},
as well as new protocols in measurement-based quantum computing~\cite{kissinger2017MBQC,Backens2020extraction,cao2022multiagent},
surface codes~\cite{horsman2017surgery,magicFactories,autoCCZ,hanks2019effective,gidney2022pair}
and other fault-tolerant architectures~\cite{litinski2022active,shaw2022quantum}.

These results in quantum computing show that diagrammatic reasoning can lead to new insights and algorithms that go beyond what is known or what even can be derived using other methods.
However, these graphical languages are in actuality not restricted to just studying quantum computing.
In fact, diagrams, the objects of a graphical calculus, can represent arbitrary tensor networks, which can represent arbitrary
$2^n$-dimensional tensors
and so they can be used for a wide variety of problems.
Whereas one would in general perform tensor contractions in order to compute with tensor networks,
a graphical calculus equips its diagrams with a formal rewrite system, which respects their tensor semantics, and allows for reasoning in terms of two-dimensional algebra.



In this work, we focus on \emph{counting problems}
which are of both practical and theoretical importance for a variety of domains, from computing partition functions in statistical mechanics~\cite{NatureComputation}, to probabilistic reasoning~\cite{ROTH1996273} and planning~\cite{BYLANDER1994165}.
The computational complexity of counting problems is of fundamental interest to computer science~\cite{papadimitriou1994computational}.
Counting problems also have a natural tensor network representation~\cite{LatorreTensorSAT},
and the complexity of computing with tensor networks has been thoroughly studied~\cite{damm_complexity_2002}.
In practice, tensor contraction algorithms for counting problems have been developed, showing competitive performance against the state of the art~\cite{kourtis_fast_2019,gray_hyper-optimized_2021}.


Graphical languages like the ZX-calculus, and its close relative, the ZH-calculus,
have been used to rederive complexity-theoretic results.
Townsend-Teague \emph{et al.} \cite{townsend-teague2021classifying} showed that the partition function of a family of Potts models, related to knot theory and quantum computation, is efficiently computable.
de Beaudrap \emph{et al.}~\cite{de_beaudrap_tensor_2021} proved graphically
that the decision version of a hard counting problem can be solved in polynomial time.
These proofs are constructive, in that they introduce algorithms in terms of rewriting strategies.
Even though this line of work recasts known results in a graphical language, such an approach is arguably more unifying and intuitive, and thus has promising potential for generalization.
Recent work by Laakkonen \emph{et al.}~\cite{laakkonen2022graphical} actually derived a \emph{novel}
complexity-theoretic result in the form of an improved runtime upper bound for counting problems. 
To obtain this result, reductions to specific counting problems were given a fully graphical treatment, to which then a known algorithm could be applied, after this algorithm was also treated graphically and generalized.

In this work, we continue building on this programme of applying graphical methods to counting. Specifically, we use the ZH-calculus to rederive various counting reductions that appear in the literature, providing a unified, and arguably simpler, presentation.
Among others, we give reductions from \sSAT to \stwoSAT, \sSATvar{Planar} and \sSATvar{Monotone}. See Table~\ref{tbl: reductions} for an overview.
Our direct proof that \stwoSAT is \sP-complete also allows us to considerably simplify the proof that computing the permanent of an integer matrix is \sP-complete.
Our results show that graphical languages can form a useful tool for the study of counting complexity.


In Section~\ref{sec:prelim} we introduce the basics of counting complexity, the ZH-calculus and how to represent \sSAT in ZH. Then in Section~\ref{sec:reductions} we present our main reductions from \sSAT by rewriting ZH-diagrams. Section~\ref{sec:evalzh} considers the converse problem of reducing ZH-diagram evaluation to \sSAT. We conclude in Section~\ref{sec:conclusion}, but note that we also present some additional reductions and proofs in the appendices.

\section{Preliminaries}\label{sec:prelim}

\subsection{Counting reductions}

    Counting complexity is defined in terms of the complexity classes \sP and $\sPk{M}$, which are the `counting analogues' of \NP. The class \sP, first defined by Valiant in 1979 \cite{valiant_complexity_1979-2}, is the class of problems which can be defined as counting the number of accepting paths to a non-deterministic Turing machine (NTM) which halts in polynomial time, whereas $\sPk{M}$ is the class of problems which can be defined as counting, modulo $M$, the number of accepting paths to an NTM (that also halts in polynomial time). Note that the notation $\pP$ is also used to indicate $\sPk{2}$. These complexity classes are clearly related to \NP, which consists of problems that can be defined as deciding whether an NTM has \emph{any} accepting path. 
    
    Famously, the Boolean satisfiability problem \textbf{SAT} is \NP-complete \cite{cook_complexity_1971}. Similarly, there are notions of \sP-completeness and $\sPk{M}$-completeness \cite{valiant_complexity_1979-2}. A problem $\mathcal{A}$ is \sP-hard ($\sPk{M}$-hard) if any problem in \sP ($\sPk{M}$) can be solved in polynomial time given an oracle for $\mathcal{A}$ (that is, there exists a Cook reduction from any problem in \sP to $\mathcal{A}$). A problem $\mathcal{A}$ is \sP-complete ($\sPk{M}$-complete) if it is both \sP-hard ($\sPk{M}$-hard) and is in \sP ($\sPk{M}$). 

    \begin{definition}
        Suppose $\phi : \mathbb{B}^n \to \mathbb{B}$ is a Boolean formula in Conjunctive Normal Form (CNF),
        \begin{equation}
            \phi(x_1, \dots, x_n) = \bigwedge_{i = 1}^{m} (c_{i1} \lor c_{i2} \lor \cdots \lor c_{ik_i})
        \end{equation}
        where $c_{ij} = x_l$ or $\lnot x_l$ for some $l$, and let $\scount{\phi} = |\{\vec{x} \mid \phi(\vec{x}) = 1\}|$. Each argument to $\phi$ is called a \emph{variable} and each term $c_{i1} \land \dots \land c_{ik_i}$ a \emph{clause}. Then, we define the following problems:
        \begin{enumerate}
            \item \textbf{SAT}: Decide whether $\scount{\phi} > 0$,
            \item \sSAT: Compute the value of $\scount{\phi}$,
            \item $\sSATk{M}$: Compute the value of $\scountk{M}{\phi} := \scount{\phi} \mod M$.  
        \end{enumerate}
    \end{definition}

    We additionally define variants, \textbf{kSAT}, \skSAT, $\skSATk{M}$, which represent the case where $\phi$ is restricted to contain only clauses of size at most $k$ (note that some sources take this to be size \emph{exactly} $k$, but we can recover this from our definition by adding dummy variables to each clause). We also take $\pSAT$ as alternate notation for $\sSATk{2}$.

    To each formula $\phi$ we associate two graphs: the \emph{incidence graph} is a bipartite graph with one vertex for each variable and one for each clause, and where a variable vertex is connected to a clause vertex if it occurs in that clause. The \emph{primal graph} has one vertex for each variable, which are connected together if the variables occur together in a clause.
    
    The Cook-Levin theorem \cite{cook_complexity_1971} shows that \kSAT is \NP-complete for $k \geq 3$, but in fact also shows that \skSAT is \sP-complete and $\skSATk{M}$ is $\sPk{M}$-complete for any $M$, as it maps any NTM into a Boolean formula such that the number of satisfying assignments is exactly equal to the number of accepting paths. We will consider variants on these problems, and specifically the case where the structure of the formula $\phi$ is restricted in some way. For each of these variants, we will append a prefix to $\textbf{SAT}$ to indicate the restriction:
    \begin{itemize}
        \item \textbf{PL:} The incidence graph of the formula is planar.
        \item \textbf{MON:} The formula is monotone - it contains either no negated variables or no unnegated variables.
        \item \textbf{BI:} The primal graph is bipartite - the variables can be partitioned into two sets such that each clause contains at most one variable from each set.
    \end{itemize}

\subsection{The ZH-calculus}

The ZH-calculus is a rigorous graphical language for reasoning about ZH-diagrams in terms of rewriting~\cite{backens_zh_2019}. We will give here a short introduction, referring the reader to~\cite[Section~8]{vandewetering2020zxcalculus} for a more in-depth explanation.

ZH-diagrams represent tensor networks \cite[Section 4.1]{orus_practical_2014} composed of the two generators of the language, the Z-spider and the H-box. The generators and their corresponding tensor interpretations are
\begin{equation}
    \input{./figures/qpl-zh-generators.tikz}
\end{equation}
where the H-box is labeled with a constant $a \in \mathbb{C}$, and we assume $a = -1$ if not given. The tensors corresponding to the generators are composed according to the tensor product and each wire connecting two tensors indicates a contraction, i.e. a summation over a common index~\cite{vandewetering2020zxcalculus}. We will also use two derived generators - the Z-spider with a phase, and the X-spider. These are given in terms of the other generators as:
\begin{equation}
    \input{./figures/qpl-zh-generators-derived.tikz}
\end{equation}
Note that the tensors are symmetric under permutation of their wires, or indices.
This implies that only the connectivity, or the topology, of the tensor network matters. In particular, we will not distinguish indices of generators as inputs and outputs as in \cite{backens_zh_2019}. Any ZH-diagram with $n$ open wires therefore represents a tensor with $n$ indices. In the special case of no open wires this represents a scalar, and we will call such diagrams scalar diagrams.

The rewriting rules of the ZH-calculus are shown in Appendix \ref{sec:zhrules}. The rules are \emph{sound}, i.e. they respect the tensor semantics, and also \emph{complete} for complex-valued linear maps, i.e. if two ZH-diagrams represent the same tensor, then there exists a sequence of rewrites which transforms one diagram to the other.

\subsection{\sSAT instances as ZH-diagrams}

    To embed \sSAT instances into ZH-diagrams, we use the translation of de Beaudrap \emph{et al.} \cite{de_beaudrap_tensor_2021} where each variable becomes a Z-spider, each clause a zero-labeled H-box, and X-spiders are used for negation. In particular the mapping is as follows
    \begin{equation}
        \input{./figures/qpl-zh-sat-defs.tikz}
    \end{equation}
    and to form \sSAT instances, we combine these as
    \begin{equation}
        \input{./figures/qpl-zh-sat-structure.tikz}
    \end{equation}
    where $G$ is a collection of wires and negations, connecting each variable to its corresponding clauses. Due to cancellation of adjacent X-spiders, an instance has an X-spider between a variable and a clause if the variable appears \emph{unnegated} in that clause, and a wire if it appears negated. For example, for the formula $\phi(x_1, x_2, x_3) = (x_1 \lor \lnot x_2 \lor \lnot x_3) \land (x_2 \lor x_3) \land (\lnot x_1 \lor \lnot x_2)$, we have:
    \begin{equation}
        \input{./figures/qpl-zh-sat-example.tikz}
    \end{equation}
    In this representation, a formula that is planar corresponds to a planar ZH-diagram and a monotone one corresponds to a ZH-diagram where there are no X-spiders or where there is an X-spider between every H-box and Z-spider. Instances with maximum clause size $k$ correspond to ZH-diagrams where every H-box has degree at most $k$.

\section{Reductions from \sSAT}\label{sec:reductions}

    We will show using the ZH-calculus that the restricted versions of \sSAT defined above---planar, monotone or bipartite--- are \sP- and/or $\sPk{M}$-complete - Table \ref{tbl: reductions} gives an overview of our reductions. All these results are already known in the literature, as will be discussed in each section, but our main contribution is to provide a simplifying and unifying viewpoint through the use of the ZH-calculus.

    \begin{table}
        \centering
        \begin{tabular}{l|l|cc|ccccc}
            \emph{Result} & \emph{Reduction} & $\sPk{k}$ & \NP & \textbf{PL-} & \stwoSAT & \textbf{MON-} & \textbf{BI-} & \textbf{3DEG-} \\
            \midrule
            Theorem \ref{sattoplsatthm} & $\sSAT \to \sSATvar{PL}$ & \checkmark & \checkmark & \checkmark & & & & \\
            Theorem \ref{satto2sat} & $\sSAT \to \stwoSAT$ & \checkmark & & \checkmark & \checkmark & & \checkmark & \\
            Theorem \ref{sattomonosattheorem} & $\sSAT \to \sSATvar{MON}$ & $\dagger$ & & \checkmark & \checkmark & \checkmark & & \\
            Theorem \ref{sattocubicsattheorem} & $\sSAT \to \sSATvar{3DEG}$ & \checkmark & & \checkmark & \checkmark & \checkmark & \checkmark & \checkmark \\
        \end{tabular}
        \caption{
            An overview of the main reductions presented in this paper. The two leftmost columns give each theorem and the corresponding reduction. The middle columns (marked $\sPk{k}$, and \NP) are given a checkmark if the corresponding reduction is valid for that complexity class as well as for \sP. A dagger is written for $\sPk{k}$ if there are some additional restrictions placed on $k$. The rightmost columns (marked \textbf{PL-}, etc) show what structure each reduction preserves - a checkmark is given if the corresponding reduction preserves the properties of the given \sSAT variant (here \stwoSAT indicates that the maximum clause size is two), i.e.~the reduction presented in Theorem~\ref{satto2sat} sends planar instances to planar instances, but does not send monotone instances to monotone instances. This applies for each complexity class that reduction is valid for (e.g Theorem \ref{satto2sat} also implies a reduction $\sSATvark{k}{PL} \to \stwoSATvark{k}{PL}$).
        }
        \label{tbl: reductions}
    \end{table}

    \subsection{$\sSAT \to \sSATvar{PL}$}  
    The first, and most commonly taught, proof that \SATvar{PL} is \NP-complete was published in 1982 by Lichtenstein \cite{lichtenstein_planar_1982}. This reduction is parsimonious - every satisfying assignment of the original formula corresponds to one satisfying assignment of the planar formula. Hence, this proves also that \sSATvar{PL} and $\sSATvark{M}{PL}$ are complete for \sP and $\sPk{M}$. Lichtenstein's construction uses a large gadget to eliminate non-planarity. In the following construction, we derive a similar gadget from first principles, by building on a famous identity from quantum computing.

    \begin{lemma}
        \label{sattoplsat}
        For any $\phi \in \skSAT$ with $n$ variables and $m$ clauses and $k \geq 3$, there is a planar $\phi' \in \skSAT$ such that $\scount{\phi} = \scount{\phi'}$. Furthermore, $\phi'$ has $O(n^2m^2)$ variables and clauses, and is computable in $O(\mathrm{poly}(n, m))$ time.
    \end{lemma}

    \begin{proof}
        Any instance $\phi \in \skSAT$ can be drawn in the plane as a ZH-diagram with some number of crossing wires. By using the famous identity that a SWAP gate can be written as the composition of 3 CNOTs, we have that \cite{nielsen_quantum_2010}:
        \begin{equation}
            \input{./figures/reductions/sat-to-planarsat-cnots.tikz}
        \end{equation}
        We now need to rewrite the X-spider, which represents a classical XOR function, into CNF for this to be a valid \sSAT instance. Unfortunately, the direct translation via the Tseytin transformation \cite{tseitin_complexity_1983} does not preserve planarity. However, we can instead use the following decomposition of an XOR as NAND gates, which is planar:
        \begin{equation}
            \input{./figures/reductions/sat-to-planarsat-xornand.tikz}
        \end{equation}
        Finally, NAND gates themselves have the following planar Tseytin transformation \cite{tseitin_complexity_1983} into CNF:
        \begin{equation}
            \input{./figures/reductions/sat-to-planarsat-nandcnf.tikz}
        \end{equation}
        Therefore, applying this to $\phi$ gives $\phi'$ with $n + 12c$ variables and $m + 36c$ clauses where $c$ is the number of crossings. If the $\phi$ is drawn with straight-line wires only, then since there are at most $nm$ wires in the diagram and each pair can cross at most once, we have $c \leq O(n^2m^2)$. As this rewrite introduces only clauses of size three or less, $\phi'$ is still a $\skSAT$ instance.
    \end{proof}

    \begin{theorem}
        \label{sattoplsatthm}
        We have the following:
        \begin{enumerate}
            \item $\skSATvar{PL}$ and $\sSATvar{PL}$ are $\sP$-complete for any $k \geq 3$.
            \item $\skSATvark{M}{PL}$ and $\sSATvark{M}{PL}$ are $\sPk{M}$-complete for any $M \geq 2$ and $k \geq 3$. 
            \item $\kSATvar{PL}$ and $\SATvar{PL}$ are $\NP$-complete for any $k \geq 3$.
        \end{enumerate}
    \end{theorem}

    \begin{proof}
        \hfill
        \begin{enumerate}
            \item This follows immediately from Lemma \ref{sattoplsat} since the size of the rewriting does not depend on the clause size.
            \item This also follows from Lemma \ref{sattoplsat}, since $\scount{\phi} = \scount{\phi'}$ implies $\scountk{M}{\phi} = \scountk{M}{\phi'}$ for any $M$. 
            \item This follows immediately from Lemma \ref{sattoplsat} since if $\scount{\phi} = \scount{\phi'}$ implies that $\phi$ is satisfiable if and only if $\phi'$ is satisfiable.\qedhere
        \end{enumerate}
    \end{proof}

    \subsection{\sSAT to \stwoSAT}

    While it is known that \stwoSAT is \sP-complete \cite{valiant_complexity_1979-1}, the proof by Valiant relies on a chain of reductions from \sSAT to the permanent of an integer matrix, to the permanent of a binary matrix, to counting perfect matchings in graphs, to counting all matchings in graphs, and then finally to \stwoSATvar{MON-BI}. Moreover, this proof does not generalize to the case of $\sPk{M}$ - in fact, proof that $\ptwoSAT$ is $\pP$-complete was only shown 27 years later in 2006 using a completely different method of holographic reductions \cite{valiant_accidental_2006}, and then a reduction for any fixed $M$ was given in 2008 by Faben \cite{faben_complexity_2008}. In this section we give a simple direct reduction from \sSAT to \stwoSAT that applies both for \sP and $\sPk{M}$. 

    \begin{lemma}[\protect{\cite[Lemma 3.3]{laakkonen2022graphical}}]
        \label{hboxpi}
        The following equivalence holds:
        \begin{equation}
            \input{./figures/dpll-neg-kclause.tikz}
        \end{equation}
    \end{lemma}

    \begin{lemma}
        \label{satto2sat2kp1}
        For any $M = 2^r + 1$ with $r \in \mathbb{N}$ and $\phi \in \sSATk{M}$ with $n$ variables and $m$ clauses, there is a $\phi' \in \stwoSATk{M}$ with $O(n + mr)$ variables such that $\scountk{M}{\phi} = \scountk{M}{\phi'}$, and $\phi'$ can be computed in $O(\mathrm{poly}(n, m, r))$ time.
    \end{lemma}

    \begin{proof}
        By evaluating the tensors, we have $\input{./figures/completeness-zhz2box.tikz}$ and therefore:
        \begin{equation}
            \input{./figures/reductions/sat-to-2sat-2kp1.tikz}
        \end{equation}
        In this way we can rewrite all of the clauses in $\phi$ to form a suitable $\phi'$.
    \end{proof}

    \begin{lemma}
        \label{satto2satfixk}
        For any $M > 2$ and $\phi \in \sSATk{M}$ with $n$ variables and $m$ clauses, there is a $\phi' \in \stwoSATk{M}$ with $O(n + mM)$ variables such that $\scountk{M}{\phi} = \scountk{M}{\phi'}$, and $\phi'$ can be computed in $O(\mathrm{poly}(n, m, M))$ time.
    \end{lemma}

    \begin{proof}
        By evaluating the tensors, we have $\input{./figures/reductions/sat-to-2sat-fixk-add.tikz}$ for all $r \in \mathbb{C}$. Therefore:
        \begin{equation}
            \input{./figures/reductions/sat-to-2sat-fixk.tikz}
        \end{equation}
        In this way we can rewrite all of the clauses in $\phi$ to form a suitable $\phi'$.
    \end{proof}
    
    \begin{theorem}
        \label{satto2sat}
        We have the following:
        \begin{enumerate}
            \item $\stwoSATk{M}$ is $\sPk{M}$-complete for any $M \geq 2$.
            \item \stwoSAT is \sP-complete.
        \end{enumerate}
    \end{theorem}

    \begin{proof}
        \hfill
        \begin{enumerate}
            \item If $M = 2$, this follows from Lemma \ref{satto2sat2kp1} with $r = 0$. If $M > 2$, then since $M$ is fixed, this follows from Lemma \ref{satto2satfixk}.
            \item For any $\phi \in \sSAT$ with $n$ variables, note that $0 \leq \scount{\phi} \leq 2^n$. Hence $\scount{\phi} = \scountk{2^n + 1}{\phi}$, and so we can apply Lemma \ref{satto2sat2kp1} with $r = n$ to generate $\phi' \in \stwoSAT$ such that $\scount{\phi} = \scountk{2^n + 1}{\phi'} = \scount{\phi'} \mod 2^n + 1$ in polynomial time, giving a polynomial-time counting reduction from $\sSAT$ to $\stwoSAT$. \qedhere
        \end{enumerate}
    \end{proof}

    \begin{corollary}
        $\stwoSATvark{M}{PL}$ is $\sPk{M}$-complete for any $M \geq 2$, and $\stwoSATvar{PL}$ is $\sP$-complete.
    \end{corollary}

    \begin{proof}
        Note that the reductions given in Lemmas \ref{satto2sat2kp1} and \ref{satto2satfixk} preserve the planarity of the input instance. Hence this follows by first applying Lemma \ref{sattoplsatthm} and then Theorem \ref{satto2sat}. 
    \end{proof}

    \begin{corollary}
        $\stwoSATvark{M}{BI}$ is $\sPk{M}$-complete for any $M \geq 2$, and $\stwoSATvar{BI}$ is \sP-complete. 
    \end{corollary}

    \begin{proof}
        When we apply Lemmas \ref{satto2sat2kp1} and \ref{satto2satfixk}, the \stwoSAT instance obtained will always be bipartite, so this follows from Theorem \ref{satto2sat}. We can see this as the primal graph has vertices in two groups: the set $V$ of vertices corresponding to variables of the original formula, and the sets $C_i$ of the vertices introduced to decompose clauses. The subgraph for each $C_i$ is clearly bipartite, so let $C_{i}^A$ and $C_{i}^B$ be the corresponding partition. Each vertex in $V$ only connects to at most one vertex $c_i$ in each $C_i$, and assume without loss of generality that $c_i \in C_{i}^A$. Then the whole graph can be partitioned as $V \cup C_{1}^{B} \cup \cdots \cup C_{m}^{B}$ and $C_{1}^A \cup \cdots \cup C_{m}^A$, so it is bipartite.
    \end{proof}

    \subsection{$\sSAT \to \sSATvar{MON}$}

    While in the previous section we showed that \stwoSAT was \sP-complete, other proofs \cite{valiant_complexity_1979-1} of this fact actually consider the subset \stwoSATvar{MON-BI}. In this section we give a reduction from \sSAT to \sSATvar{MON}, allowing us to remove negations from any CNF formula. This shows that our graphical method is not any less powerful than the reduction via the permanent, and we argue that this chain of reductions is more intuitive because it allows us to gradually restrict the formulae, rather than jumping straight to a highly restrictive variant.

    \begin{lemma}
        \label{sattomonosat}
        For any $r \geq 0$ and $\phi \in \sSATk{2^r}$ with $n$ variables, $m$ clauses, and maximum clause size at least two, there is a monotone $\phi' \in \sSATk{2^r}$ with $O(n + nmr)$ variables and $O(m + nmr)$ clauses such that $\scountk{2^r}{\phi} = \scountk{2^r}{\phi'}$. Additionally, $\phi'$ preserves the maximum clause size of $\phi$, and can be computed in $O(\mathrm{poly}(n, m, r))$ time.
    \end{lemma}

    \begin{proof}
        By evaluating the tensors, we have $\input{./figures/reductions/sat-to-monosat-2box.tikz}$ and therefore,
        \begin{equation}
            \input{./figures/reductions/sat-to-monosat.tikz}
        \end{equation}
        where the first equality follows from the Tseytin transformation of the NOT gate \cite{tseitin_complexity_1983}. Thus we can remove every negation in $\phi$ as follows:
        \begin{equation}
            \input{./figures/reductions/sat-to-monosat-final.tikz}
        \end{equation}
        There are at most $nm$ negations in $\phi$, and each can be rewritten with $O(r)$ clauses and variables. Note that this rewrite introduces only clauses of size two, so the maximum clause size is preserved.
    \end{proof}

    \begin{theorem}
        \label{sattomonosattheorem}
        We have the following:
        \begin{enumerate}
            \item $\skSATvark{2^r}{MON}$ and $\sSATvark{2^r}{MON}$ are $\sPk{2^r}$-complete for any $r \geq 0$ and $k \geq 2$. 
            \item $\skSATvar{MON}$ and $\sSATvar{MON}$ are $\sP$-complete for any $k \geq 2$.
        \end{enumerate}
    \end{theorem}

    \begin{proof}
        \hfill
        \begin{enumerate}
            \item This follows immediately from Lemma \ref{sattomonosat}. Since the transformation preserves maximum clause size, this holds for either bounded or unbounded clause size.
            \item For any $\phi \in \skSAT$ with $n$ variables and $k \geq 2$, note that $0 \leq \scount{\phi} \leq 2^n$. Hence $\scount{\phi} = \scountk{2^{n + 1}}{\phi}$, and so we can apply Lemma \ref{sattomonosat} with $r = n + 1$ to generate $\phi' \in \skSATvar{MON}$ such that $\scount{\phi} = \scountk{2^{n + 1}}{\phi'} = \scount{\phi'} \mod 2^{n + 1}$ in polynomial time, giving a polynomial-time counting reduction from $\skSAT$ to $\skSATvar{MON}$. The same argument also gives a reduction for $\sSATvar{MON}$. \qedhere
        \end{enumerate}
    \end{proof}

    \subsection{Other Reductions}

    Using similar methods we can also consider other restrictions of \sSAT. For example, in Appendix \ref{sec:combinedsat}, we combine Theorems \ref{sattoplsatthm}, \ref{satto2sat}, and \ref{sattomonosattheorem} with further reductions to show that \stwoSATvar{MON-BI-PL-3DEG}, where \textbf{3DEG-} indicates each variable participates in at most three clauses, is \sP-complete. This case is interesting because it is as small as possible - if we instead have each variable participate in only two clauses then this is in \textbf{P} \cite{dahllof_counting_2002}. Indeed most upper bounds on runtime for \stwoSAT have better special case for this type of formula \cite{wahlstrom_tighter_2008, dahllof_counting_2002}. 
    
    As phase-free ZH-diagrams naturally encode \sSAT instances, the ZH-calculus is mostly suited to treat variations on the \sSAT problem. To apply the technique of graphical reasoning to other (counting) problems, we hence may need to use other graphical calculi. In particular, in Appendix \ref{sec:perfmatch}, we show how the ZW-calculus \cite{hadzihasanovic_diagrammatic_2015} is naturally adapted to both the \textbf{\#XSAT} problem (of which \textbf{\#1-in-3SAT} is a special case) and the \textbf{\#PERFECT-MATCHINGS} problem, and use this shared structure to give graphical reductions showing that both are \sP-complete. This complements the recent result of Carette \textit{et al.}~illustrating with the ZW-calculus that \textbf{\#PLANAR-PERFECT-MATCHINGS} is in \textbf{P} \cite{carette_compositionality_2023}. 

    While the reductions given above contribute to simplifying the literature in their own right, we can also derive other simplifications from them. For example, the original proof by Valiant \cite{valiant_complexity_1979-2} (and the simplification by Ben-Dor and Halevi \cite{ben-dor_zero-one_1993}) that computing the permanent of a boolean matrix is \sP-complete relies on a reduction from \sthreeSAT. It would be simpler to reduce from \stwoSATvar{MON-}, but the original proof that \stwoSATvar{MON-} is \sP-complete relies on a reduction from the permanent, so this would be circular. 
    
    However, by Theorems \ref{satto2sat} and \ref{sattomonosattheorem} we have \sP-completeness for \stwoSATvar{MON-} independent of the permanent. This then allows us to give an alternate, simpler proof that computing the permanent of an integer matrix is \sP-complete, which we present in Appendix~\ref{sec:intperm}. 
    In the original proof of Ben-Dor and Halevi \cite{ben-dor_zero-one_1993}, they construct for a given \sthreeSAT instance a weighted directed graph with two cycles per variable and a gadget of seven vertices for each clause such that the permanent of the adjacency matrix of the graph equals the value of the \sthreeSAT instance. As finding a suitable clause gadget was difficult, they found a suitable one using computer algebra. 
    Our proof adapts theirs, but as we can start from a \stwoSATvar{MON-} instance, our graph can be made simpler, only requiring one cycle per variable, and a symmetric clause gadget of just four vertices. This was found, and can easily be proven correct, by hand.

\section{Evaluating Scalar ZH-Diagrams}
    \label{sec:evalzh}

    While we have so far shown that variants of \sSAT can be embedded into ZH-diagrams, and thus that the problem of evaluating an arbitrary scalar ZH-diagram is \sP-hard, we haven't yet answered how much harder it might be. I.e~whether this problem is in \sP. In this section we will show that evaluating scalar ZH-diagrams comprised of a certain fragment of generators is complete for $\textbf{FP}^{\sP}$. \textbf{FP} is the class of functions that can be evaluated in polynomial time by a deterministic Turing machine (i.e the function analog of \textbf{P}), and $\textbf{FP}^{\sP}$ is thus the class of functions that can be evaluated in polynomial time by a deterministic Turing machine with access to an oracle for a \sP-complete problem (in our case we will use \sSAT).

    In order to consider the problem of evaluating scalar ZH-diagrams formally, we first define the problem $\mathrm{Eval\mhyphen}F$ which is the task of finding the complex number corresponding to a scalar diagram that exists in fragment $F$ of a graphical calculus. A fragment is a set of diagrams built from arbitrary combinations of a fixed subset of generators.

    \begin{definition}
        For a given fragment $F$, the problem $\mathrm{Eval\mhyphen}F$ is defined as follows:
        \begin{description}
            \item[Input] A scalar diagram $D \in F$ consisting of $n$ generators and wires in total, where any parameters of the generators of $D$ can be expressed in $O(\mathrm{poly}(n))$ bits.
            \item[Output] The value $D \in \mathbb{C}$.  
        \end{description}
        The runtime of an algorithm for $\mathrm{Eval\mhyphen}F$ is defined in terms of the parameter $n$. 
    \end{definition}

    We will examine two fragments $\mathrm{ZH}_{\pi/2^k} \supseteq \mathrm{ZH}_{\pi}$ and show that they can be reduced to \sSAT. Far from being purely academic, $\mathrm{ZH}_{\pi}$ is expressive enough to capture Toffoli-Hadamard quantum circuits, and $\mathrm{ZH}_{\pi/2^k}$ can additionally capture Clifford+T quantum circuits, both of which are approximately universal for quantum computation.

    \begin{definition}
        $\mathrm{ZH}_{\pi}$ is the fragment of ZH-calculus given by the following generators:
        \begin{equation}
            \input{./figures/completeness-zhzdef.tikz}
        \end{equation}
    \end{definition}

    \begin{lemma}
        \label{completenesstseytinxor}
        The following diagram equivalence holds:
        \begin{equation}
            \input{./figures/completeness-xor3sat.tikz}
        \end{equation}
        This is derived from the Tseytin transformation of the XOR operation \cite{tseitin_complexity_1983}.
    \end{lemma}

    \begin{theorem}
        \label{completenesszhzcomplete}
        There is a polynomial-time counting reduction from the problem $\mathrm{Eval\mhyphen}\mathrm{ZH}_{\pi}$ to \sthreeSAT and so $\mathrm{Eval\mhyphen}\mathrm{ZH}_{\pi}$ is in $\textbf{FP}^\sP$. Note that $\mathrm{Eval\mhyphen}\mathrm{ZH}_{\pi}$ is equivalent to the problem $\sSAT_\pm$ as defined in \cite{laakkonen2022graphical}.
    \end{theorem} 

    \begin{proof}
        In order to rewrite a diagram $D$ from $\mathrm{ZH}_{\pi}$ into \sthreeSAT, we first rewrite all of the non-scalar H-boxes into zero H-boxes with two legs:
        \begin{equation}
            \input{./figures/completeness-zhzsplit.tikz}
        \end{equation}
        Where the second equality follows from Lemma A.4 in \cite{laakkonen2022graphical}. Now, we can remove all the X-spiders and $\pi$-phase Z-spiders as follows, to rewrite into a valid \sSAT diagram:
        \begin{enumerate}
            \item Any spiders or H-boxes with no legs should be removed from the diagram. Evaluate them by concrete calculation and multiply their values together to get a scalar multiplier $c$ for the diagram. If there are no such spiders or H-boxes, set $c = 1$.
            \item Extract the phases from all $\pi$-phase Z-spiders as follows:
            \begin{equation}
                \input{./figures/completeness-piunfuse.tikz}
            \end{equation}
            \item Unfuse the phase of every X-spider with at least two legs:
            \begin{equation}
                \input{./figures/completeness-xunfuse.tikz}
            \end{equation}
            \item For any X-spiders with at least three legs, unfuse them and apply Lemma \ref{completenesstseytinxor} to each X-spider: \begin{equation}
                \input{./figures/completeness-xorsplit.tikz}
            \end{equation}
            \item Replace X-spiders with one leg as follows: \begin{equation}
                \input{./figures/completeness-zbasis.tikz}
            \end{equation}
            \item Excepting the single Z-spider with a $\pi$-phase, use the $SF_Z$ rule to fuse Z-spiders wherever possible. If there are two two-legged zero H-boxes connected together directly, introduce a Z-spider between them using the $I_Z$ rule. 
        \end{enumerate}
        At this point, this diagram follows the form of a \sthreeSAT diagram except for a possible single $\pi$-phase Z-spider. To remove the $\pi$-phase Z-spider, we can write the diagram as a sum of two diagrams which don't contain this phase:
        \begin{equation}
            \input{./figures/completeness-negsum.tikz}
        \end{equation}
        Since these two diagrams are \sthreeSAT diagrams associated with some Boolean functions $f_1$ and $f_2$, the value of $D$ is then given by $D = c(\scount{f_1} - \scount{f_2})$.
    \end{proof}

    Note that this method of splitting an instance into positive and negative components is similar to Bernstein and Vazirani's proof that $\textbf{BQP} \subseteq \textbf{P}^{\sP}$ \cite[Theorem 8.2.5]{bernstein_quantum_1997}.  Now we can move on to considering larger fragments fairly easily. We will show that by using gadgets to copy certain ``magic states'' - that is, one-legged H-boxes with specific labels which we can split as sums - we can introduce phases which are multiples of $\frac{\pi}{2^k}$ for any fixed $k$.

    \begin{definition}
        $\mathrm{ZH}_{\pi/2^k}$ for $k \in \mathbb{N}$ is the fragment of ZH-calculus given by the following generators
        \begin{equation}
            \input{./figures/completeness-zhqidef.tikz}
        \end{equation}
        where $n \in \mathbb{Z}$.
    \end{definition}
    
    \begin{lemma}
        \label{completenesszhqicopy}
        The following diagram equivalence holds:
        \begin{equation}
            \input{./figures/completeness-zhqicopy.tikz}
        \end{equation}
    \end{lemma}

    \begin{proof} We have that
        \begin{equation}
            \input{./figures/completeness-zhqicopy-proof.tikz}
        \end{equation}
        where the first equality follows from \cite[Lemma 3.2]{kuijpers_graphical_2019}.
    \end{proof}
    
    \begin{lemma}
        \label{completenesszhqicomplete}
        For all $k > 0$, there is a reduction from $\mathrm{Eval\mhyphen}\mathrm{ZH}_{\pi/2^k}$ to $\mathrm{Eval\mhyphen}\mathrm{ZH}_{\pi/2^{k - 1}}$.
    \end{lemma}
    
    \begin{proof}
        Let $D$ be a diagram in $\mathrm{ZH}_{\pi/2^k}$. In order to rewrite this into a diagram in $\mathrm{ZH}_{\pi/2^{k - 1}}$, we need to remove all Z-spiders with phases that are odd multiples $n$ of $\frac{\pi}{2^k}$, since even multiples of $\frac{\pi}{2^k}$ are already valid for $\mathrm{ZH}_{\pi/2^{k - 1}}$. Assume that $n = 2m + 1$, then:
        \begin{equation}
            \input{./figures/completeness-zhk-phasesplit.tikz}
        \end{equation}
        Let $a = e^{\frac{i\pi}{2^k}}$, then by applying Lemma \ref{completenesszhqicopy}, fold up all of $a$-labeled H-boxes into one
        \begin{equation}
            \input{./figures/completeness-zhqicopyimpl.tikz}
        \end{equation}
        and note that $\input{./figures/completeness-zhk-a2.tikz}$ is in $\mathrm{ZH}_{\pi/2^{k - 1}}$. Remove all scalar H-boxes from $D$ and let their product be $c$. Finally, we can split the remaining $a$-labeled H-box, giving $D$ as the sum of two diagrams $D_1$ and $D_2$ in $\mathrm{ZH}_{\pi/2^{k - 1}}$:
        \begin{equation}
            \input{./figures/completeness-zhk-boxsplit.tikz}
        \end{equation}
        so then $D = c(D_1 + aD_2)$.
    \end{proof}

    \begin{theorem}
        $\mathrm{Eval\mhyphen}\mathrm{ZH}_{\pi/2^k}$ is $\textbf{FP}^\sP$-complete for all $k \geq 0$.
    \end{theorem}

    \begin{proof}
       By induction on Lemma \ref{completenesszhqicomplete} with the base case $k = 0$ given by Theorem \ref{completenesszhzcomplete}, we have that $\mathrm{Eval\mhyphen}\mathrm{ZH}_{\pi/2^k}$ is in $\textbf{FP}^\sP$. It is clearly $\textbf{FP}^\sP$-hard, as $\mathrm{ZH}_{\pi/2^k}$ contains the diagrams representing \sSAT instances, and hence is $\textbf{FP}^\sP$-complete.
    \end{proof}

    The fragment $\mathrm{ZH}_{\pi/2^2}$ captures precisely those diagrams that represent postselected Clifford+T quantum circuits. Our results above hence also lead to a proof that $\textbf{PostBQP} \subseteq \textbf{P}^{\sP}$, although note that this is weaker than Aaronson's result that $\textbf{PostBQP} = \textbf{PP}$~\cite{aaronson2005quantum}.




\section{Conclusion}\label{sec:conclusion}

In this paper we have used the ZH-calculus to simplify and unify the proofs of several known results in counting complexity. In particular, we examined various variants of \sSAT and show that they are \sP-complete, and similarly that the corresponding $\sSATk{M}$ variants are $\sPk{M}$-complete. We for instance produced a simple direct reduction from \sSAT to \stwoSAT, which considerably simplified existing proofs that proceed via a reduction to the matrix permanent. Our results show that graphical calculi like the ZH-calculus, even though originally meant for the domain of quantum computing, can provide an intuitive framework for working with counting problems, through their interpretation as tensor networks. 

We also briefly examined how other graphical calculi can be used to reason about other counting problems, especially the ZW-calculus and its connection to counting perfect matchings in graphs. A natural future direction is to explore which counting problems can be naturally formulated in a graphical calculus. Finally, we also observed how the original domain of quantum computing can be related to \sSAT via the ZH-calculus, and show that the computational problem of evaluating scalar ZH-diagrams that represent postselected Clifford+T or Toffoli+Hadamard quantum circuits is in $\textbf{FP}^\sP$, and hence can be efficiently evaluated with an \sSAT oracle - evaluating whether this leads to a more efficient method for simulating quantum circuits is an interesting avenue for future research.

\subsection*{Acknowledgments}

Some of this work was done while TL was a student at the University of Oxford, and the
results in Section \ref{sec:evalzh} are also presented in his Master's thesis \cite{laakkonen_graphical_2022}. We thank Richie Yeung, Matty Hoban, Julien Codsi, and the anonymous QPL reviewers for helpful feedback.

\bibliographystyle{eptcs}
\bibliography{bibliography}

    \appendix

    \section{Rewriting Rules}
    \label{sec:zhrules}

    The ZH-calculus is equipped with the following set of sound and complete rewriting rules \cite{backens_zh_2019}:
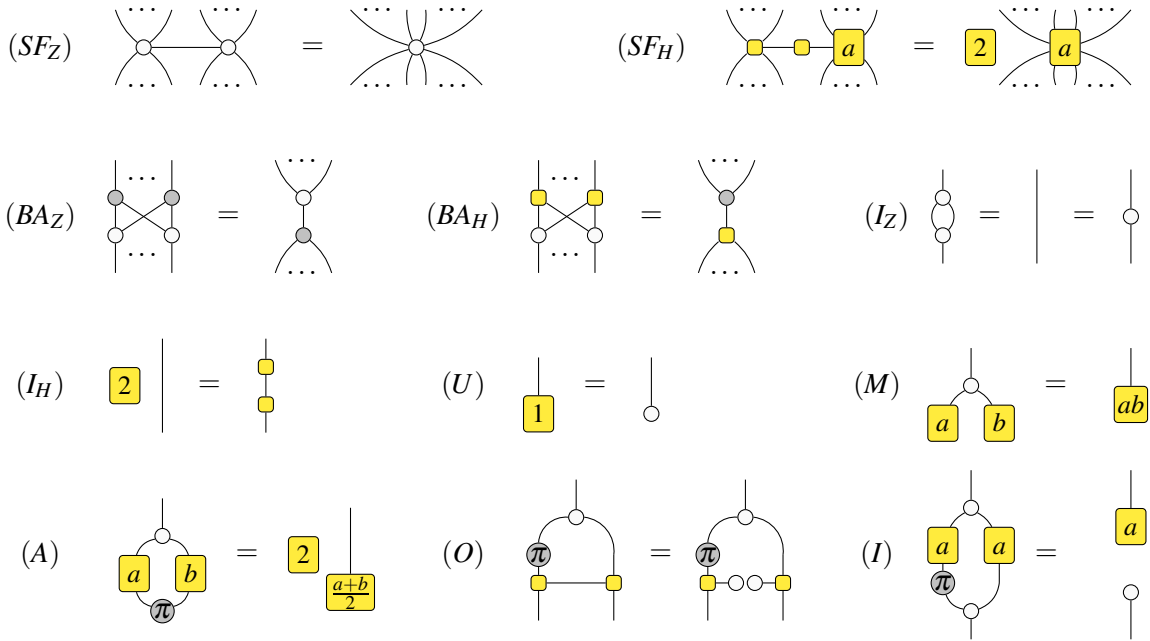
    \begin{equation*}
        \input{./figures/qpl-zhcalculusrules.tikz}
    \end{equation*}
    We will also use the following derived rewriting rules \cite[Lemmas 2.10-2.24]{backens_completeness_2021},
    \begin{equation}
        \input{./figures/qpl-zhcalculusrules-extra.tikz}
    \end{equation}
    as well as the generalized $(M)$ rule \cite[Lemma 2.3]{backens_zh_2019}:
    \begin{equation}
        \input{./figures/qpl-zhcalculusrules-extra2.tikz}
    \end{equation}

    \section{$\sSAT \to \sSATvar{MON-BI-PL-3DEG}$}
    \label{sec:combinedsat}

    The smallest subset of \stwoSAT that has been considered in the literature is \stwoSATvar{PL-MON-BI-CUBIC}, where \textbf{CUBIC-} indicates that the primal graph of the instance is 3-regular \cite{xia_3-regular_2006}. In this section, we show that \sSATvar{3DEG} is \sP-complete graphically, by relating the zero-labeled H-box with the Fibonacci numbers. Here, \textbf{3DEG-} indicates that every variable appears in at most three clauses. Then in Theorem \ref{combined2sat} we combine all of our reductions to show \sP-completeness for \stwoSATvar{PL-MON-BI-3DEG}. This is slightly less restrictive than \stwoSATvar{PL-MON-BI-CUBIC}, but retains the interesting property that the maximum degree is the lowest possible (if the maximum degree was two, then the problem can be solved in polynomial time \cite{dahllof_counting_2002}), while avoiding the complicated global construction of the original proof.

    We will make use of the following identities concerning the Fibonacci numbers, defined by:
    \begin{equation}
        F_{n} = F_{n - 1} + F_{n - 2} ~~~~ F_{1} = 1 ~~~~ F_{0} = 0
    \end{equation}

    \begin{lemma}
        \label{fibcoprime}
        $\gcd(F_{n}, F_{n - 1}) = 1$ for all $n > 1$.
    \end{lemma}

    \begin{proof}
        $\gcd(F_{n}, F_{n - 1}) = \gcd(F_{n - 1} + F_{n - 2}, F_{n - 1}) = \gcd(F_{n - 2}, F_{n - 1})$ thus by induction $\gcd(F_{n}, F_{n - 1}) = \gcd(F_2, F_1) = 1$.
    \end{proof}

    \begin{lemma}[{\cite{halton_divisibility_1966}}]
        \label{fibexists}
        For every $M \geq 1$ there exists some $0 < n \leq M^2$ such that $F_n \equiv 0 \mod M$.
    \end{lemma}

    \begin{lemma}[{\cite{halton_divisibility_1966}}]
        \label{fibclosedform}
        $F_n = \left\lfloor \frac{\phi^n}{\sqrt{5}} + \frac{1}{2}\right\rfloor$ where $\phi = \frac{1 + \sqrt{5}}{2}$ is the golden ratio.
    \end{lemma}

    \begin{lemma}
        \label{sattocubicsatident}
        Suppose that $F_n \equiv 0 \mod M$, then the following rewrite holds:
        \begin{equation}
            \input{./figures/reductions/sat-to-cubicsat-ident.tikz}
        \end{equation}
    \end{lemma}

    \begin{proof}
        Note that the Fibonacci numbers are defined by
        \begin{equation}
            \begin{pmatrix}
                F_{n + 1} \\ F_{n}
            \end{pmatrix} = \begin{pmatrix}
                1 & 1 \\ 1 & 0
            \end{pmatrix}\begin{pmatrix}
                F_{n} \\ F_{n - 1}
            \end{pmatrix} \implies \begin{pmatrix}
                1 & 1 \\ 1 & 0
            \end{pmatrix}^n = \begin{pmatrix}
                F_{n + 1} & F_{n} \\ F_{n} & F_{n - 1}
            \end{pmatrix}
        \end{equation}
        and so, supposing that $F_n \equiv 0 \mod M$, we have
        \begin{equation}
            \input{./figures/reductions/sat-to-cubicsat-ident-proof.tikz}
        \end{equation}
        but by Lemma \ref{fibcoprime}, $F_{n - 1}$ is coprime to $F_n$, so it must be coprime to $M$ and thus invertible.
    \end{proof}

    \begin{lemma}
        \label{sattocubicsatlemma}
        Given $M$ and $k$ such that $F_k \equiv 0 \mod M$, for any $\phi \in \sSATk{M}$ with $n$ variables, $m$ clauses, and maximum clause size at least two, there is a $\phi' \in \sSATk{M}$ with $O(nmk)$ variables and $O(nmk)$ clauses such that every variable has degree at most three, and $\scountk{M}{\phi} = c \cdot \scountk{M}{\phi'} \mod M$ where $c$ is computable in $O(\mathrm{poly}(n, m, k))$ time. Additionally, $\phi'$ can be computed in $O(\mathrm{poly}(n, m, k))$ time and preserves the maximum clause size of $\phi$.
    \end{lemma}

    \begin{proof}
        For every variable in $\phi$ with degree more than three, we can apply the following rewrite by Lemma \ref{sattocubicsatident}:
        \begin{equation}
            \input{./figures/reductions/sat-to-cubicsat.tikz}
        \end{equation}
        since $d \leq m$, this adds at most $O(nmk)$ clauses and variables. Note that we took H-box sequences of length $2k$ (i.e two applications of Lemma \ref{sattocubicsatident}) in order to preserve any bipartite structure in $\phi$. Thus $\phi$ now has degree at most three. Since these rewrites only add clauses of size two, they preserve the maximum clause size. The total scalar factor accumulated across all of these rewrites is
        \begin{equation}
            c = \prod_{i=1}^n \begin{cases} F_{k - 1}^{6 - 2d_i} & d_i > 3 \\ 1 & d_i \leq 3 \end{cases}
        \end{equation}
        where $d_i$ is the degree of variable $i$, which is clearly computable in polynomial time.
    \end{proof}

    \begin{theorem}
        \label{sattocubicsattheorem}
        We have the following:
        \begin{enumerate}
            \item $\skSATvark{M}{3DEG}$ and $\sSATvark{M}{3DEG}$ are $\sPk{M}$-complete for any $M \geq 1$ and $k \geq 2$. 
            \item $\skSATvar{3DEG}$ and $\sSATvar{3DEG}$ are $\sP$-complete for any $k \geq 2$. 
        \end{enumerate}
    \end{theorem}

    \begin{proof}
        \hfill
        \begin{enumerate}
            \item By Lemma \ref{fibexists}, there exists some $k$ such that $F_k \equiv 0 \mod M$ and it is easy to compute. Therefore, this follows from Lemma \ref{sattocubicsatlemma} since the maximum clause size is preserved.
            \item Let $\psi \in \sSAT$ have $n$ variables, and $\phi$ be the golden ratio. Pick $k = \left\lceil\log_\phi\left((2^n + 1)\sqrt{5} - \frac{1}{2}\right)\right\rceil = O(n)$, then $F_k \geq 2^n + 1$ by Lemma \ref{fibclosedform} and $\scount{\phi} = \scountk{F_k}{\phi}$. Apply Lemma \ref{sattocubicsatlemma} with $M = F_k$ to generate $\phi' \in \sSATvar{3DEG}$ such that $\scount{\phi} = c\cdot \scountk{F_k}{\phi'} \mod F_k = c \cdot \scount{\phi'} \mod F_k$ for some polynomial-time computable $c$. The same argument applies for bounded clause size. \qedhere
        \end{enumerate}
    \end{proof}

    \begin{lemma}
        \label{2sattobi2satlemma}
        Given $r \geq 0$, for any $\phi \in \stwoSATk{2^r + 1}$ with $n$ variables and $m$ clauses, there is a bipartite $\phi' \in \stwoSATk{2^r + 1}$ with $O(nmr)$ variables and $O(nmr)$ clauses such that $\scountk{2^r + 1}{\phi} = \scountk{2^r + 1}{\phi'}$, which can be computed in $O(\mathrm{poly}(n, m, r))$ time, and that preserves monotonicity and planarity.
    \end{lemma}

    \begin{proof}
        Note that we have the following rewriting rule
        \begin{equation}
            \input{./figures/reductions/2sat-to-bi2sat-lemma.tikz}
        \end{equation}
        where the last line follows from the proof of Lemma \ref{satto2sat2kp1}. After this is applied to every clause of size two, every path between two vertices in the incidence graph will have even length (since every edge is replaced by four edges), and hence the incidence graph is bipartite.
    \end{proof}

    \begin{theorem}
        \label{combined2sat}
        We have the following:
        \begin{enumerate}
            \item $\stwoSATvar{MON-BI-PL-3DEG}$ is $\sP$-complete
            \item $\ptwoSATvar{MON\mhyphen BI\mhyphen PL\mhyphen 3DEG}$ is $\pP$-complete
        \end{enumerate}
    \end{theorem}

    \begin{proof}
        We can reduce from arbitrary $\sSAT$ instances to $\stwoSATvar{MON-BI-PL-3DEG}$ by applying the previously given reductions in the following order:
        \begin{align*}
            \sSAT \stackrel{\ref{sattoplsat}}{\to} \sSATvar{PL} &\stackrel{\ref{satto2sat}}{\to} \stwoSATvar{PL} \stackrel{\ref{sattomonosat}}{\to} \stwoSATvar{MON-PL} \\
            &\stackrel{\ref{2sattobi2satlemma}}{\to} \stwoSATvar{MON-BI-PL} \stackrel{\ref{sattocubicsatlemma}}{\to} \stwoSATvar{MON-BI-PL-3DEG}
        \end{align*}
        In both cases, we first reduce to $\sSATvar{PL}$ using Lemma \ref{sattoplsat}. Then, we continue differently:
        \begin{enumerate}
            \item Apply Theorem \ref{satto2sat} and Theorem \ref{sattomonosattheorem} to reduce to $\stwoSATvar{MON-PL}$, as these preserve planarity. Then, given $\phi$ with $n$ variables, we apply Lemma \ref{2sattobi2satlemma} with $r = n + 1$ to obtain $\phi' \in \stwoSATvar{MON-BI-PL}$ such that $\scount{\phi} = \scountk{2^r}{\phi'} = \scount{\phi'} \mod 2^r$. Then reduce $\phi'$ to \stwoSATvar{MON-BI-PL-3DEG} using Theorem \ref{sattocubicsattheorem}, since it preserves planarity, monotonicity and bipartite structure.
            \item Similarly, given $\phi \in \sSAT$, apply Lemma \ref{satto2sat2kp1} with $r = 0$, Lemma \ref{sattomonosat} with $r = 1$, Lemma \ref{2sattobi2satlemma} with $r = 0$, and Lemma \ref{sattocubicsatlemma} with $M = 2$ and $k = 3$, to obtain $\phi' \in \stwoSATvar{MON-BI-PL-3DEG}$ such that $\scountk{2}{\phi} = \scountk{2}{\phi'}$. \qedhere
        \end{enumerate}
    \end{proof}

    \section{\textbf{\#PERFECT-MATCHINGS} and the ZW-Calculus}
    \label{sec:perfmatch}
    
    While so far we have worked exclusively with the ZH-calculus, which naturally represents \sSAT, we can use other calculi to attack other problems. In this section, we will use the ZW-calculus to examine the connection between the problems \textbf{\#XSAT} and \textbf{\#PERFECT-MATCHINGS}, and sketch an argument that they are both \sP-complete. Like the connection between \sSAT and \stwoSAT given using the ZH-calculus, with this technique we can circumvent the usual reduction via the permanent. 

    \begin{definition}
        Let $f : \mathbb{B}^n \to \mathbb{B}$ be a boolean function defined by
        \begin{equation}
            f(x) = \bigwedge_{i = 0}^m \phi(c_{i1}, \dots, c_{ik_i})
        \end{equation}
        where $c_{ij} = x_k$ or $\lnot x_k$ for some $k$, and $\phi(\vec{x}) = 1$ if and only if $w(\vec{x}) = 1$, where $w(\vec{x})$ is the Hamming weight of $\vec{x}$. Each $\phi$ term defines a clause, and so $f(x) = 1$ iff every clause has exactly one true literal. The problem \textbf{\#XSAT} is to compute $\scount{f}$. When $k_i = 3$ for all $i$, this is also known as \textbf{\#1-in-3SAT}.
    \end{definition}

    \begin{definition}
        The problem \textbf{\#PERFECT-MATCHINGS} is as follows: given an undirected simple graph $G$, compute the number of perfect matchings of $G$. That is, the number of independent edge sets of $G$ that cover each vertex exactly once. We denote this quantity $\mathrm{PerfMatch}(G)$.
    \end{definition}
    
    The ZW-calculus is a graphical calculus built from two generators, W-spiders and Z-spiders, which are flexsymmetric \cite{hadzihasanovic_diagrammatic_2015}. The Z-spider is a close analogue of the Z-spider in ZH-calculus (except with a $\pi$ phase), whereas the W-spider represents the W-state:
    \begin{equation}
        \input{./figures/reductions/zw-calculus-gens.tikz}
    \end{equation}
    These, along with wires, caps, and cups, are combined with tensor product and tensor contraction in the same way as in the ZH-calculus. Like the ZH-calculus, we will treat diagrams purely as tensor networks rather than formal objects - hence equality of diagrams is just equality of tensors. This calculus is also equipped with a set of sound and complete rewrite rules, including the following spider fusion rules,
    \begin{equation}
        \input{./figures/reductions/perfmatch-zwfusion.tikz}
    \end{equation}
    as well as others which we omit for brevity as we don't use them explicitly here. In the same way that ZH-calculus diagrams naturally represent \sSAT instances with Z-spiders and clauses, the ZW-calculus naturally represents \textbf{\#XSAT} instances with the following mapping:
    \begin{equation}
        \input{./figures/reductions/zw-calculus-xsat.tikz}
    \end{equation}

    \begin{lemma}
        \label{xsat-sp-complete}
        \textbf{\#XSAT} is \sP-complete
    \end{lemma}

    \begin{proof}
        We can translate from \stwoSAT to \textbf{\#XSAT} with the following correspondence:
        \begin{equation}
            \input{./figures/reductions/xsat-sat-correspond.tikz}
        \end{equation}
        In the other direction, we can translate \textbf{\#XSAT} to \sSAT by first expanding every clause of size more than three as follows:
        \begin{equation}
            \input{./figures/reductions/xsat-1in3sat-split.tikz}
        \end{equation}
        Then, each clause
        \begin{equation}
            \begin{aligned}
                \phi(x, y, z) &= (x \land \lnot y \land \lnot z) \lor (\lnot x \land y \land \lnot z) \lor (\lnot x \land \lnot y \land z) \\
                \phi(x, y) &= (x \land \lnot y) \lor (\lnot x \land y) \\
                \phi(x) &= x
            \end{aligned}
        \end{equation}
        can be rewritten as a bounded number of CNF clauses.
    \end{proof}

    However, as Carette \textit{et al.} \cite{carette_compositionality_2023} note, diagrams of the ZW-calculus can also naturally represent instances of \textbf{\#PERFECT-MATCHINGS} by taking each vertex of the graph to be a W-spider and edges of the graph as wires. Therefore, any ZW-diagram containing only W-spiders represents an instance of \textbf{\#PERFECT-MATCHINGS}.

    \begin{theorem}
        \label{perfmatch-sp-complete}
        \textbf{\#PERFECT-MATCHINGS} is \sP-complete.
    \end{theorem}

    \begin{proof}
        Suppose we are given an instance $f$ of \textbf{\#XSAT} on $n$ variables as a ZW-diagram, then to transform it to an instance of \textbf{\#PERFECT-MATCHINGS} we need to remove all Z-spiders, which represent variables. First split all the variables so they have degree at most three:
        \begin{equation}
            \input{./figures/reductions/perfmatch-zsplit.tikz}
        \end{equation}
        Then we can use the following rewrites to remove all variables with degree two and three:
        \begin{equation}
            \label{perfmatch-zthree}
            \input{./figures/reductions/perfmatch-zthree.tikz}
        \end{equation}
        We are left with only variables of degree one, some extraneous Z-spiders with degree two, and a constant factor of $2^{-c}$. To complete the reduction to \textbf{\#PERFECT-MATCHINGS} it then remains to show we can get rid of these degree-one variables and the degree-two Z-spiders. 

        We can remove the Z-spiders by considering the whole diagram modulo $2^{n + c} + 1$: since we started with a \textbf{\#XSAT} instance with $n$ variables, we have $0 \leq \scount{f} \leq 2^n$, so the value of the remaining diagram is at most $2^{n + c}$. It is hence sufficient to calculate modulo $2^{n + c} + 1$ for our resulting diagram. In this setting, we have
        \begin{equation}
            \input{./figures/reductions/perfmatch-zident.tikz} = \begin{pmatrix}1 & 0 \\ 0 & -1\end{pmatrix} \equiv  \begin{pmatrix}1 & 0 \\ 0 & 2^{n + c}\end{pmatrix} = \begin{pmatrix}1 & 0 \\ 0 & 2\end{pmatrix}^{n + c} = \input{./figures/reductions/perfmatch-zident2.tikz} \pmod{2^{n + c} + 1}
        \end{equation}
        since
        \begin{equation}
            \input{./figures/reductions/perfmatch-square.tikz} = \begin{pmatrix}1 & 0 \\ 0 & 2\end{pmatrix}
        \end{equation}
        and hence we are left with a diagram containing only variables of degree one, and no other Z-spiders. To remove these variables of degree one, note that
        \begin{equation}
            \input{./figures/reductions/perfmatch-copyz.tikz}
        \end{equation}
        and thus we can combine all the variables of degree one together:
        \begin{equation}
            \input{./figures/reductions/perfmatch-uncopy.tikz}
        \end{equation}
        Finally, we can remove the last variable by splitting the diagram as a sum of diagrams, neither of which contain any Z-spiders:
        \begin{equation}
            \input{./figures/reductions/perfmatch-zsum.tikz}
        \end{equation}
        Thus, these two diagrams each represent an instance of \textbf{\#PERFECT-MATCHINGS} - let us denote the graphs of the corresponding instances as $G_1$ and $G_2$. The construction above allows us to obtain $G_1$ and $G_2$ in polynomial time, and we have
        \begin{equation}
            \scount{f} = 2^{-c}(\mathrm{PerfMatch}(G_1) + \mathrm{PerfMatch}(G_2)) \mod{2^{n + c} + 1}
        \end{equation}
        hence \textbf{\#PERFECT-MATCHINGS} is \sP-hard. We can also see that $\textbf{\#PERFECT-MATCHINGS} \in \sP$, since we can use Equation \eqref{perfmatch-zthree} to rewrite all wires into variables of degree two, and thus transform an instance of \textbf{\#PERFECT-MATCHINGS} into an instance of \textbf{\#XSAT}. Therefore, \textbf{\#PERFECT-MATCHINGS} is \sP-complete.
    \end{proof}

    \section{\sP-Completeness for the Permanent}
    \label{sec:intperm}

    The proof by Valiant \cite{valiant_complexity_1979-2} that the permanent of an integer-valued matrix, $\mathbb{Z}\textbf{-Permanent}$, is \sP-complete, and the simplified proof by Ben-Dor and Halevi \cite{ben-dor_zero-one_1993}, both rely on a reduction from \sthreeSAT. This reduction could be simplified by using \stwoSAT instead, but this was unfortunately not possible, as the proof that \stwoSAT is \sP-complete relies itself on a reduction from the permanent. However, since we proved in Theorem \ref{satto2sat} that \stwoSAT is \sP-complete independent of the permanent, we can make use of this to simplify the reduction for $\mathbb{Z}\textbf{-Permanent}$ further. In this section, we detail this reduction, which shows that $\mathbb{Z}\textbf{-Permanent}$ is \sP-hard. Our construction and proof is essentially identical to that of Ben-Dor and Halevi \cite{ben-dor_zero-one_1993}, with the exception that we start with an instance of \stwoSATvar{MON}, and can hence use simpler gadgets. 
    
    \begin{definition}
        Given a directed edge-weighted graph $G$ with edges $E$, a \emph{cycle-cover} of $G$ is a set $E' \subseteq E$ of simple cycles that partition the vertices of $G$. Note that self-loops are permitted in $G$. The weight of a cycle cover is the product of all the weights of the edges in $E'$.
    \end{definition}

    \begin{lemma}
        Let $G$ be a directed graph with self-loops and edge weights $w_{ij}$, and let $A$ be its adjacency matrix, i.e.~$A_{ij} = w_{ij}$ if $i$ and $j$ are connected or $A_{ij} = 0$ otherwise. Then the permanent of $A$ is the sum of weights of all cycle-covers of $G$. We denote this number by $\scount{G}$.
    \end{lemma}

    Given this, it is sufficient to reduce \stwoSATvar{MON} to the problem of determining the sum of weights of cycle-covers of a graph. We aim to construct a graph $G_\phi$ from a \stwoSATvar{MON} instance $\phi$ with $n$ variables and $m$ clauses, in polynomial time, such that $\scount{G_\phi} = f(\phi) \cdot \scount{\phi}$ for some easily computable and suitably bounded $f(\phi)$. We construct $G_\phi$ as follows:
    \begin{enumerate}
        \item For each variable $x_i$ in $\phi$, introduce a vertex $v_i$ to $G_\phi$.
        \item For each clause in $\phi$, introduce a \emph{clause gadget} of four vertices to $G_\phi$, the structure of which we will describe momentarily. Two of the vertices of this gadget are designated as the first and second input respectively.
        \item For every variable vertex, add a self-loop $s_i$ of weight one to $G_\phi$.
        \item For every variable vertex $v_i$, add an edge of weight one from $v_i$ to an unused input $c_1$ of the first clause in which $v$ appears. Then add an edge of weight one from $c_1$ to $c_2$, an unused input of the second clause in which $v$ appears. Continue similarly until all $d_i$ clauses in which $v$ appears have been processed, then add an edge of weight one from $c_{d_i}$ to $v_i$. Let these edges be labeled as $c_{ij}$.
    \end{enumerate}
    An example of this construction is given below, the loop of edges proceeding from each variable highlighted in a different color:
    \begin{equation}
        \input{./figures/reductions/permanent-2sat-example.tikz}
    \end{equation}
    The clause gadget is given by the following graph
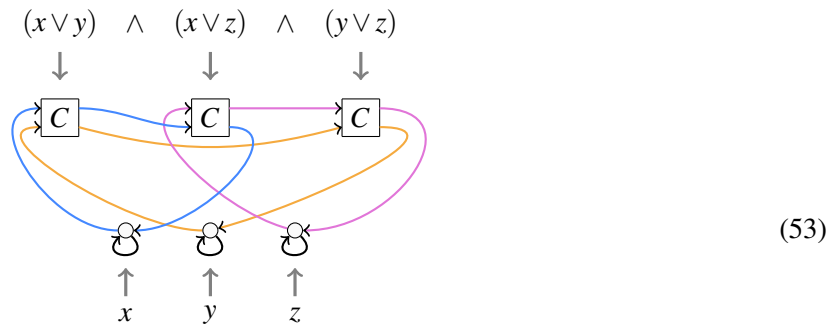
    \begin{equation*}
        \input{./figures/reductions/2sat-perm-gadget.tikz}
    \end{equation*}
    where the top-most vertex is the first input, and the bottom-most vertex is the second. Bidirectional edges represent a pair of edges, one in each direction, with the same weights. Now let $E_c$ be the set of edges in $G_\phi$ that are internal to clause gadgets, and let $E_r = E \setminus E_c$ be the rest.

    \begin{definition}
        Let a \emph{partial cover} of $G_\phi$ be a subset $E_r' \subseteq E_r$. A \emph{completion} of $E_r'$ is a cycle cover of $G_\phi$ given by $E_r' \cup E_c'$  where $E_c' \subseteq E_c$. We call the \emph{weight} of $E_r'$ the sum of the weights of all completions of $E_r'$.  
    \end{definition}

    Let us say that a partial cover $E_r'$ is induced by a satisfying assignment $\vec{x}$ of $\phi$ if, for every variable $x_i$ assigned false in $\vec{x}$, $s_i \in E_r'$ and $c_{ij} \notin E_r'$ for all $1 \leq j \leq d_i$, and for every variable $x_i$ assigned true in $\vec{x}$, $c_{ij} \in E_r'$ for all $1 \leq j \leq d_i$ and $s_i \notin E_r'$.  We wish to argue that the weight of a partial cover of $G_\phi$ is non-zero if and only if it is induced by a satisfying assignment.

    \begin{lemma}
        \label{2satpermlemma}
        Let $E_r'$ be a partial cover of $G_\phi$, then the weight of $E_r'$ is $4^m$ if $E_r'$ is induced by a satisfying assignment, and zero otherwise. Moreover, each such $E_r'$ is induced by a unique satisfying assignment.
    \end{lemma}

    \begin{proof}
        Suppose $E_r'$ is induced by a satisfying assignment. Then for each clause gadget, the ingoing and outgoing edges are included in $E_r'$ for either one or both inputs (otherwise there is an unsatisfied clause). The possible completions of $E_r'$ are as follows for each clause gadget:
        \begin{equation}
            \input{./figures/reductions/2sat-perm-gadget-cycles-1.tikz}
        \end{equation}
        The dotted edges represent edges not included in the cycle-cover. Then the total weight of each clause gadget over the completions is four in either case, so the overall weight of $E_r'$ is $4^m$. Now suppose that $E_r'$ is not induced by a satisfying assignment. Note that if the number of incoming and outgoing edges of each clause gadget in $E_r'$ is not equal, then the weight of $E_r'$ is zero, as there is no valid completion of $E_r'$ (because there can be no such cycle-cover). Therefore, the only remaining case is that there is at least one clause gadget which has no incoming and outgoing edges, or has one incoming edge and one outgoing edge on the opposing input (otherwise $E_r'$ would be induced by a satisfying assignment). In either case, we can see the total weight of the gadget over the completions is zero:
        \begin{equation}
            \input{./figures/reductions/2sat-perm-gadget-cycles-2.tikz}
        \end{equation} 
        Hence, the weight of $E_r'$ must also be zero. Note that each $E_r'$ that is induced by a satisfying assignment must be induced by a unique assignment, since you can recover the assignment from $E_r'$.
    \end{proof}

    Clearly, the sum of weights of all partial covers of $G_\phi$ is the same as the sum of weights of all cycle-covers of $G_\phi$. But by Lemma \ref{2satpermlemma},  this is $4^m \scount{\phi}$, so $\scount{\phi} = 4^{-m} \scount{G_\phi}$, and thus $\mathbb{Z}\textbf{-Permanent}$ is \sP-hard, as $G_\phi$ can be computed in polynomial time from $\phi$. 

    It is interesting to note that the constructions of Ben-Dor and Halevi, and Valiant, both make use of negative-weight edges and have $f(\phi) = k^m$ for some even integer $k$ ($12$ for Ben-Dor and Halevi, and $4^5$ for Valiant). In order to further simplify the next steps of the reduction to $\mathbb{B}\textbf{-Permanent}$, it would be desirable to have no negative weights, or $k = 1$. However, as Valiant points out \cite{valiant_complexity_1979-2}, neither of these is likely to be possible:
    \begin{itemize}
        \item If $k$ is odd, then $\scount{\phi} \equiv \scount{G_\phi} \mod 2$, but $\scount{G_\phi} \mod 2$ is easy to compute \cite{ben-dor_zero-one_1993} (as the parity of the permanent is equal to the parity of the determinant), so then $\textbf{P} = \pP$ and $\NP = \textbf{RP}$ by the Valiant-Vazirani theorem \cite{valiant_np_1986}.
        \item Suppose $G_\phi$ is constructed by reduction from \textbf{3SAT}. If there are no negative-weighted edges, then the existence of any cycle-cover of $G_\phi$ indicates the existence of a satisfying assignment to $\phi$. But determining if a cycle-cover exists is easy for general directed graphs, so then $\textbf{P} = \textbf{NP}$.
    \end{itemize}
    This last argument does not hold up for our construction, since we start from $\stwoSATvar{MON}$, for which it is trivial to determine if a satisfying assignment exists (indeed, one always exists by setting every variable true). However, we can still rule out the possibility of a reduction without negative-weighted edges: it is known that $\mathbb{Z}\textbf{-Permanent}$ with non-negative weights has an FPRAS \cite{jerrum_polynomial-time_2004}, whereas \stwoSATvar{MON} does not, unless $\NP = \textbf{RP}$ \cite[Theorem 57]{welsh_complexity_2011}.

\end{document}

%% file: quantum.tikzstyles
\tikzstyle{gate}=[shape=rectangle, text height=1.5ex, text depth=0.25ex, yshift=0.5mm, fill=white, draw=black, minimum height=3mm, yshift=-0.5mm, minimum width=3mm, font={\small}, tikzit category=circuit, inner sep=2pt]
\tikzstyle{big gate}=[shape=rectangle, text height=1.5ex, text depth=0.25ex, yshift=0.5mm, fill=white, draw=black, minimum height=10mm, yshift=-0.5mm, minimum width=5mm, font={\small}, tikzit category=circuit]
\tikzstyle{Z dot}=[inner sep=0mm, minimum size=2mm, shape=circle, draw=black, fill={rgb,255: red,221; green,255; blue,221}, tikzit category=zx]
\tikzstyle{Z phase dot}=[minimum size=5mm, font={\footnotesize\boldmath}, shape=rectangle, rounded corners=2mm, inner sep=0.2mm, outer sep=-2mm, scale=0.8, tikzit shape=circle, draw=black, fill={rgb,255: red,221; green,255; blue,221}, tikzit draw=blue, tikzit category=zx]
\tikzstyle{X dot}=[Z dot, shape=circle, draw=black, fill={rgb,255: red,255; green,136; blue,136}, tikzit category=zx]
\tikzstyle{X phase dot}=[Z phase dot, tikzit shape=circle, tikzit draw=blue, fill={rgb,255: red,255; green,136; blue,136}, font={\footnotesize\boldmath}, tikzit category=zx]
\tikzstyle{hadamard}=[fill=yellow, draw=black, shape=rectangle, inner sep=0.6mm, minimum height=1.5mm, minimum width=1.5mm, tikzit category=zx]
\tikzstyle{paulibox}=[fill={rgb,255: red,221; green,221; blue,255}, draw=black, shape=rectangle, inner sep=0.6mm, minimum height=5mm, minimum width=5mm, font={\footnotesize}, text height=1.5ex, text depth=0.25ex, tikzit category=zx]
\tikzstyle{vertex}=[inner sep=0mm, minimum size=1mm, shape=circle, draw=black, fill=black, tikzit category=misc]
\tikzstyle{vertex set}=[inner sep=0mm, minimum size=1mm, shape=circle, draw=black, fill=white, font={\footnotesize\boldmath}, tikzit category=misc]
\tikzstyle{small black dot}=[fill=black, draw=black, shape=circle, inner sep=0pt, minimum width=1.2mm, tikzit category=circuit]
\tikzstyle{cnot ctrl}=[fill=black, draw=black, shape=circle, inner sep=0pt, minimum width=1.2mm, tikzit category=circuit]
\tikzstyle{cnot targ}=[fill=white, draw=white, shape=circle, tikzit category=circuit, label={center:$\oplus$}, inner sep=0pt, minimum width=2.1mm, tikzit fill={rgb,255: red,102; green,204; blue,255}, tikzit draw=black]
\tikzstyle{ket}=[fill=white, draw=black, shape=regular polygon, regular polygon sides=3, regular polygon rotate=-30, scale=0.7, inner sep=1pt, tikzit category=circuit, tikzit shape=rectangle, tikzit fill=green]
\tikzstyle{bra}=[fill=white, draw=black, shape=regular polygon, regular polygon sides=3, regular polygon rotate=30, scale=0.7, inner sep=1pt, tikzit category=circuit, tikzit shape=rectangle, tikzit fill=red]
\tikzstyle{scalar}=[shape=rectangle, text height=1.5ex, text depth=0.25ex, yshift=0.5mm, fill=white, draw=black, minimum height=5mm, yshift=-0.5mm, minimum width=5mm, font={\small}]
\tikzstyle{clabel}=[fill=white, draw=none, shape=rectangle, tikzit fill={rgb,255: red,56; green,255; blue,242}, font={\footnotesize}, inner sep=1pt, tikzit category=labels]
\tikzstyle{empty diagram}=[draw={gray!40!white}, dashed, shape=rectangle, minimum width=1cm, minimum height=1cm, tikzit category=misc]
\tikzstyle{amap}=[fill=white, draw=black, shape=NEbox, tikzit category=asymmetric, tikzit fill=yellow, tikzit shape=rectangle]
\tikzstyle{amap conj}=[fill=white, draw=black, shape=NWbox, tikzit category=asymmetric, tikzit fill=green, tikzit shape=rectangle]
\tikzstyle{amap adj}=[fill=white, draw=black, shape=SEbox, tikzit category=asymmetric, tikzit fill=red, tikzit shape=rectangle]
\tikzstyle{amap trans}=[fill=white, draw=black, shape=SWbox, tikzit category=asymmetric, tikzit fill=orange, tikzit shape=rectangle]
\tikzstyle{astate}=[fill=white, draw=black, shape=NEtriangle, tikzit category=asymmetric, tikzit shape=circle, tikzit fill=yellow]
\tikzstyle{astate conj}=[fill=white, draw=black, shape=NWtriangle, tikzit category=asymmetric, tikzit shape=circle, tikzit fill=green]
\tikzstyle{astate adj}=[fill=white, draw=black, shape=SEtriangle, tikzit category=asymmetric, tikzit shape=circle, tikzit fill=red]
\tikzstyle{astate trans}=[fill=white, draw=black, shape=SWtriangle, tikzit category=asymmetric, tikzit shape=circle, tikzit fill=orange]
\tikzstyle{white dot}=[inner sep=0mm, minimum size=2mm, shape=circle, draw=black, fill={rgb,255: red,250; green,250; blue,250}]
\tikzstyle{white phase dot}=[minimum size=5mm, font={\footnotesize\boldmath}, shape=rectangle, rounded corners=2mm, inner sep=0.2mm, outer sep=-2mm, scale=0.8, tikzit shape=circle, draw=black, fill={rgb,255: red,250; green,250; blue,250}, tikzit draw=blue]
\tikzstyle{hbox}=[shape=rectangle, text height=2mm, fill={rgb,255: red,255; green,235; blue,61}, draw=black, minimum height=2mm, minimum width=2mm, font={\small}, tikzit category=zh, inner sep=0pt, rounded corners=0.5mm]
\tikzstyle{Z dot (zh)}=[inner sep=0mm, minimum size=2mm, shape=circle, draw=black, fill={rgb,255: red,250; green,250; blue,250}, tikzit category=zh]
\tikzstyle{X dot (zh)}=[Z dot, shape=circle, draw=black, fill={rgb,255: red,193; green,193; blue,193}, tikzit category=zh]
\tikzstyle{triangle}=[fill={rgb,255: red,255; green,136; blue,136}, draw=black, shape=isosceles triangle, isosceles triangle apex angle=60, minimum size=2.5mm, inner sep=0mm]
\tikzstyle{labelled hbox}=[shape=rectangle, text height=1.75ex, text depth=0.5ex, fill={rgb,255: red,255; green,235; blue,61}, draw=black, minimum height=3mm, minimum width=4mm, font={\small}, tikzit category=zh, inner sep=1.3pt, rounded corners=0.5mm]
\tikzstyle{Z phase dot (zh)}=[Z phase dot, tikzit shape=circle, tikzit draw=blue, fill={rgb,255: red,250; green,250; blue,250}, font={\footnotesize\boldmath}, tikzit category=zh]
\tikzstyle{X phase dot (zh)}=[Z phase dot, tikzit shape=circle, tikzit draw=blue, fill={rgb,255: red,193; green,193; blue,193}, font={\footnotesize\boldmath}, tikzit category=zh]
\tikzstyle{W node}=[fill=black, draw=black, shape=regular polygon, regular polygon sides=3, minimum size=2mm]
\tikzstyle{Z dot (zw)}=[fill=white, draw=black, shape=circle, minimum width=1.2mm, inner sep=0pt]
\tikzstyle{Z phase dot XL}=[Z phase dot, fill={rgb,255: red,250; green,250; blue,250}, draw=black, shape=circle, tikzit draw={rgb,255: red,191; green,0; blue,64}, tikzit shape=circle, font={\large\boldmath}, inner sep=0.0mm]

\tikzstyle{hadamard edge}=[-, dashed, dash pattern=on 2pt off 0.5pt, thick, draw={rgb,255: red,68; green,136; blue,255}]
\tikzstyle{box edge}=[-, dashed, dash pattern=on 2pt off 0.5pt, thick, draw={rgb,255: red,203; green,192; blue,225}]
\tikzstyle{brace edge}=[-, tikzit draw=blue, decorate, decoration={brace,amplitude=1mm,raise=-1mm}]
\tikzstyle{diredge}=[->, thick]
\tikzstyle{double edge}=[-, double, shorten <=-1mm, shorten >=-1mm, double distance=2pt]
\tikzstyle{gray edge}=[-, {gray!60!white}]
\tikzstyle{pointer edge}=[->, very thick, gray]
\tikzstyle{boldedge}=[-, line width=1.6pt, shorten <=-0.17mm, shorten >=-0.17mm]
\tikzstyle{bidir edge}=[<->, very thick, draw={rgb,255: red,191; green,191; blue,191}]
\tikzstyle{purple edge}=[->, thick, draw={rgb,255: red,225; green,117; blue,216}]
\tikzstyle{green edge}=[->, thick, draw={rgb,255: red,167; green,231; blue,137}]
\tikzstyle{orange edge}=[->, thick, draw={rgb,255: red,245; green,170; blue,63}]
\tikzstyle{blue edge}=[->, thick, draw={rgb,255: red,68; green,136; blue,255}]
\tikzstyle{any edge}=[->, thick, draw=cyan]
\tikzstyle{red edge}=[->, thick, draw={rgb,255: red,255; green,136; blue,136}]
\tikzstyle{bidiredge}=[<->, thick]
\tikzstyle{dashed diredge}=[->, dashed, dash pattern=on 1pt off 0.5pt]
\tikzstyle{bidashed diredge}=[<->, dashed, dash pattern=on 1pt off 0.5pt]

%% file: figures/reductions/perfmatch-zident.tikz
\begin{tikzpicture}
	\begin{pgfonlayer}{nodelayer}
		\node [style=none] (0) at (-1, 0) {};
		\node [style=none] (1) at (1, 0) {};
		\node [style={Z dot (zw)}] (2) at (0, 0) {};
	\end{pgfonlayer}
	\begin{pgfonlayer}{edgelayer}
		\draw (2) to (1.center);
		\draw (2) to (0.center);
	\end{pgfonlayer}
\end{tikzpicture}